\documentclass[10 pt, letterpaper]{article}
\usepackage{makeidx}
\usepackage{amssymb}
\usepackage{amsmath}
\usepackage{graphicx}

\newcommand{\tran}{^{\top}}

\newcommand{\prob}{{\mathbb P}}

\newcommand{\inv}{^{-1}}

\newcommand{\beq}{\begin{equation}}
\newcommand{\eeq}{\end{equation}}
\newcommand{\bea}{\begin{eqnarray}}
\newcommand{\eea}{\end{eqnarray}}
\newcommand{\beas}{\begin{eqnarray*}}
\newcommand{\eeas}{\end{eqnarray*}}
\newcommand{\ba}{\begin{array}}
\newcommand{\ea}{\end{array}}
\newcommand{\bit}{\begin{itemize}}
\newcommand{\eit}{\end{itemize}}
\newcommand{\ben}{\begin{enumerate}}
\newcommand{\een}{\end{enumerate}}

\newcommand{\ped}[1]{ _{ {\mathrm{#1} } }}

\newcommand{\Real}[1]{ { {\mathbb R}^{#1} } }

\newcommand{\bR}{{\mathbb R}}

\newtheorem{theorem}{Theorem}[section]
\newtheorem{proof}{Proof}[section]
\newtheorem{algorithm}{Algorithm}[section]
\newtheorem{assumption}{Assumption}

\newtheorem{proposition}{Proposition}[section]

\newtheorem{remark}{Remark}[section]

\newcommand{\calA}{{\cal A}}

\newcommand{\calP}{{\cal P}}

\newcommand{\calV}{ {\cal V} }

\newcommand{\tbar}{{\bar{t}}}

\begin{document}
\title{Robust Model Predictive Control via Scenario Optimization\thanks{This manuscript is a preprint of a paper accepted for publication in the IEEE Transactions on Automatic Control, with DOI: 10.1109/TAC.2012.2203054, and is subject to IEEE copyright. The copy of record will be available at http://ieeexplore.ieee.org.}
\thanks{This research has received funding from the European Union Seventh Framework
Programme (FP7/2007-2013) under grant agreement n. PIOF-GA-2009-252284 -
Marie Curie project ``ICIEMSET'' and from the Italian Ministry of University and Research under PRIN 20087W5P2K grant.}}

\author{G.C. Calafiore\thanks{Dip. di Automatica e Informatica, Politecnico di Torino,
Italy. E-mail: giuseppe.calafiore@polito.it.}$\,$ and Lorenzo Fagiano\thanks{Dip. di Automatica e Informatica, Politecnico di Torino,
Italy, and  Dept. of Mechanical Engineering, University of California at Santa Barbara, USA. E-mail: lorenzo.fagiano@polito.it.}}

\maketitle
\noindent\textbf{{Abstract} --
This paper discusses   a novel probabilistic approach for the design of robust model predictive control (MPC) laws for discrete-time linear systems affected by   parametric uncertainty and additive disturbances. The proposed technique is based on the iterated solution, at each step, of a  finite-horizon optimal control problem (FHOCP)  that takes into account a suitable number of randomly extracted scenarios of uncertainty and disturbances, followed by a  specific command selection rule implemented in a receding horizon fashion.  The scenario FHOCP is always convex, also when the uncertain parameters and disturbance belong to non-convex sets, and irrespective of how the model uncertainty influences the system's matrices. Moreover, the computational complexity of the proposed approach does not depend on the uncertainty/disturbance dimensions, and scales quadratically with the control horizon.
%Feasibility of the scenario FHOCP is guaranteed by using soft constraints and a slack variable, that can be used to monitor %the extent of the (possible) violation of the original constraints.
The main result in this paper  is related to the analysis
of the closed loop system under receding-horizon implementation of the scenario FHOCP, and essentially states that
the devised control law guarantees constraint satisfaction at each step with some a-priori assigned probability $p$, while the system's state reaches the target set either  asymptotically, or in finite time with probability at least $p$.
The proposed method may be a valid alternative when other existing techniques, either deterministic or stochastic, are not directly usable due to excessive conservatism or to numerical intractability  caused by lack of convexity of the robust or chance-constrained optimization problem.}\\
$\,$\\

\section{Introduction}\label{S:intro}
In Model Predictive Control (MPC), at each sampling time $t$, a plant's  control input $u_t\in\mathbb{R}^m$ is computed by solving a constrained finite horizon optimal control problem (FHOCP), according to a receding horizon (RH) strategy, see, e.g., \cite{MRRS00}. MPC has received an ever-increasing attention in the last decades, mainly due to the possibility of taking into account input and state constraints explicitly in the control design. The study of robust MPC approaches, able to guarantee stability and constraint satisfaction also in the presence of uncertainty/disturbances, is still a very active research area. For the case of linear time invariant (LTI) discrete time models, an extensive literature has been developed, considering the presence of either model uncertainty or external disturbances, see, e.g.,  \cite{KoBM96}-\cite{CACJ??}.
Most of the existing approaches are deterministic and aim to optimize a worst-case performance index, while enforcing constraints for all possible outcomes of the uncertainty \cite{KoBM96}-\cite{MuBF06} or disturbance \cite{ScMa98}-\cite{GoKM06}. These techniques guarantee that the designed control law is able to cope with the considered uncertainty. However, they rely on the assumption of convexity of the optimization problem, not only with respect to the control input, but also with respect to the uncertain parameters and disturbances. Moreover, the computational complexity of deterministic approaches typically grows with the complexity of the model set. In a  recent and   active research direction, stochastic MPC techniques have also been studied, see, e.g., \cite{BeBe09}-\cite{CACJ??} and the references therein. Stochastic MPC techniques exploit some known statistical description of the uncertain parameters and/or of the disturbance (e.g., the  probability distribution, or the first and second moments), %. In this way, the cost function becomes a random variable and the MPC algorithm aims to minimize its expectation. Yet,
yet they still employ deterministic algorithms and, in order to maintain tractability of the optimization problem, they typically assume that the system matrices are either perfectly known, or they have a particular structure that preserves convexity.\\
We propose here a new randomized method for robust MPC design, which is able to deal with both model uncertainty and additive disturbances. Similar to stochastic MPC techniques, we exploit information on the statistics of the uncertain parameters and disturbances. However, we do not  assume convexity or even connectedness of the model set or of the disturbance set. Still, the optimization problem in our approach is always convex, and the control law is able to robustly enforce constraints and trajectory convergence, with a probability higher than a user-defined value $p$.
Furthermore, for a given value of $p$, the computational complexity of our approach is completely independent of the complexity of the model set.
The key point enabling to achieve these features is a shift of paradigm, from a \emph{deterministic} algorithm to a
\emph{randomized} one, i.e., an algorithm that relies on  random choices in the course of its execution (see, e.g., \cite{TeCD05}).
Indeed, a key step in our main algorithm (Algorithm~\ref{Alg:MPCS}) is the solution of a scenario FHOCP, in which
we do not consider all possible outcomes of uncertainty and disturbances, but only a finite number $M$
of randomly chosen instances of them,
named the ``scenarios.'' A randomized approach for MPC has been studied also in \cite{Batina2004,KaML09}, by using a Monte Carlo technique. However, Monte Carlo approaches may be very computationally demanding and can not handle in a straightforward way the presence of state constraints. Randomization has been used also in \cite{BOBW10}, in the context of chance-constrained MPC.
However, in \cite{BOBW10} there is no guideline on how to choose $M$ in order to have the guarantee that the probability of success is at least $p$ (which is instead one of the features of the present approach) and, moreover, the resulting optimization problem is a mixed-integer linear program.
On the contrary, the approach proposed here, named MPCS (MPC via Scenario optimization), exploits relatively recent results in Random Convex Programming (RCP, see \cite{CaCa05}-\cite{Cala10})
to provide an explicit link between $M$ and $p$. Moreover, we introduce a slack variable, the ``constraint violation,'' which renders the scenario FHOCP always feasible, and that can be used to monitor the extent of the (possible) violation of the involved constraints.
Further, we show how scenario optimization can be embedded in a  receding horizon scheme, in order to provide a feedback controller that gives probabilistic guarantees of robust stability and constraint satisfaction. The approach here proposed shall be particularly interesting in all those cases where the assumptions underpinning the existing deterministic or stochastic approaches for robust MPC are not met; for example, when the dependence of the system matrices on the uncertain parameters is not affine.

\section{Problem formulation}\label{S:problem}

Consider the following uncertain, discrete time LTI  model:
\begin{equation}\label{E:model}
x_{t+1}=A(\theta)x_t+B(\theta)u_t+B_\gamma(\theta)\gamma_t
\end{equation}
where $t\in\mathbb{Z}$ is the discrete time variable, $x_t\in\mathbb{R}^n$ is the system state, $u_t\in\mathbb{R}^m$ is the control input,  $\gamma_t\in\Gamma\subset\mathbb{R}^{m_\gamma}$ is an unmeasured disturbance vector,
$\theta\in\Theta\subseteq\mathbb{R}^g$ is the vector of uncertain parameters, and $A(\theta),\,B(\theta),\,B_\gamma(\theta)$ are matrices of suitable dimensions. Let us consider the following assumptions:

\begin{assumption}\label{A:unc_bound}(Uncertainty description)
The sets $\Gamma$ and $\Sigma\doteq\{A(\theta),\,B(\theta),\,B_\gamma(\theta):\theta\in\Theta\}$ are bounded.
We assume $\gamma_t$ and $\theta$ to have stochastic nature, and we
 let $\prob_\theta$ denote the probability measure
on $\Theta$, and $\prob_\gamma$ the probability measure
on $\Gamma$. Variables $\theta$ and $\gamma_t$ are independent.
Moreover,
$\gamma\doteq \{\gamma_0,\gamma_1,\ldots\}$ is an independent identically distributed (i.i.d.) sequence
and we let $\prob_\gamma^\infty$ denote the probability measure on this sequence.
\hfill$\blacksquare$
\end{assumption}

\begin{assumption}\label{A:syst_stab}(Robust stabilizability)
The pair $A(\theta),B(\theta)$ is stabilizable for any $\theta\in\Theta$.\hfill$\blacksquare$
\end{assumption}

The control problem is to regulate the system state to a neighborhood of the origin, subject to  (possibly uncertain) input and state constraints $x_t\in\mathbb{X}(\theta),\,u_t\in\mathbb{U}(\theta),\,\forall t$. The next assumption characterizes the constraint sets.

\begin{assumption}\label{A:constr}(State and input constraints)
For any $\theta\in\Theta$, the sets $\mathbb{X}(\theta)\subset\mathbb{R}^n$ and $\mathbb{U}(\theta)\subset\mathbb{R}^m$ are convex;
 they contain the origin in their interiors and they are representable by: $
\mathbb{X}(\theta)=\left\{x\in\mathbb{R}^n:f_X(x,\theta)\preceq0\right\},$\\$
\mathbb{U}(\theta)=\left\{u\in\mathbb{R}^m:f_U(u,\theta)\preceq0\right\}$,
where $\preceq$ denotes element-wise inequalities, each entry of the functions $f_X:\mathbb{R}^n\times\Theta\rightarrow\mathbb{R}^r,\,f_U:\mathbb{R}^m\times\Theta\rightarrow\mathbb{R}^q$ is convex in $x$ and $u$, respectively, and $r,\,q$ are suitable integers.
\hfill$\blacksquare$
\end{assumption}

The parameter $\theta$ has been included in the constraints to account for practical applications where, for example, a convex function of the states (e.g., energy, load) has to be limited below some threshold, and some parameter in the function or the threshold itself depend on uncertain physical quantities (e.g., maximal energy, breaking load).
Assumptions \ref{A:unc_bound} and \ref{A:constr} are quite mild, since $\Theta$ may be unbounded and of any form,  no assumption on $\Sigma,\,\Gamma$ is made except for boundedness, no restrictions on how the parameter $\theta$ influences matrices $A(\theta),\,B(\theta),\,B_\gamma(\theta)$  are imposed, as long as the system is stabilizable, and finally no assumption on the shape of the convex sets $\mathbb{X}(\theta),\,\mathbb{U}(\theta)$ (e.g., polytopic, ellipsoidal, ...) for given $\theta\in\Theta$ is made. Mixed constraints of the form $(x,u)\in\mathbb{X}_U(\theta)$, where $\mathbb{X}_U(\theta)\subseteq\mathbb{R}^n\times\mathbb{R}^m$ is a convex set, are not considered here for simplicity, but they can be straightforwardly included in our problem settings. Due to the presence of the generally non-zero unmeasured disturbance $\gamma_t$, regulation of the system state to the equilibrium $\overline{x}=0,\,\overline{u}=0$ can not be attained. Rather, we can require regulation  to a neighborhood of the origin, described by a terminal set, which is robustly positively invariant under a terminal control law.
\begin{assumption}\label{A:terminal_set}(Terminal set and terminal control law)
A convex set $\mathbb{X}_f$ and a linear terminal control law $u=K_f\,x,\,K_f\in\mathbb{R}^{m\times n}$, exist for system (\ref{E:model}), such that $
\mathbb{X}_f = \{x: f_{X_f}(x) \preceq 0\};\,
A(\theta)x+B(\theta)K_f\,x+B_\gamma(\theta)\gamma\in\mathbb{X}_f,\,\forall \theta\in\Theta,\, \forall \gamma\in\Gamma,\,\forall x\in\mathbb{X}_f$; finally
$f_X(x,\theta)\preceq0,\,f_U(K_f\,x,\theta)\preceq0,\,\forall\theta\in\Theta,\,\forall x\in\mathbb{X}_f$,
where $f_{X_f}:\mathbb{R}^n\rightarrow\mathbb{R}^l$ has convex components, and $l$ is a suitable integer.\hfill$\blacksquare$
\end{assumption}
A possible method for constructing a terminal set and a terminal control law satisfying Assumption~\ref{A:terminal_set} is to apply results in quadratic stability and rejection of bounded disturbances for uncertain LTI systems, see, e.g., \cite{Mao03,PNTN06} and the references therein. Moreover, there is a number of contributions in the literature concerned with the computation of approximations of the (minimal) robust positively invariant terminal set $\mathbb{X}_f$, see e.g. \cite{BlMi08,RaBa10} and the references therein. In the rest of this note, we parameterize the control input as:
\beq
u_t = K_f x_t + v_t,
\label{eq:feedbacku}
\eeq
where $K_f$ is the terminal control law of Assumption \ref{A:terminal_set} (which is assumed to be known and given), and $v_t$ is a control correction to be designed. Plugging (\ref{eq:feedbacku}) into (\ref{E:model}), we obtain the discrete-time model
\begin{equation}\label{E:modeltc}
x_{t+1}=A\ped{cl}(\theta)x_t+B(\theta)v_t+B_\gamma(\theta)\gamma_t,
\end{equation}
with $A\ped{cl}(\theta) = A(\theta) +B(\theta) K_f$, which will be the basis of our developments.

\section{The Scenario-based Finite-Horizon Optimal Control Problem}\label{SS:randomized-control}
Suppose that, at a given time instant $t$,  the state $x_t$
of system (\ref{E:modeltc}) is observed. We consider the problem of determining a corrective
control sequence on a  horizon of $N$  instants forward in time. To this end,  we build a randomized finite-horizon optimal control problem (FHOCP), as described next.
Let $N$ be the chosen horizon length, and let $v_{j|t}$, $ j=0,1,\ldots,N-1$, be the $N$ predicted control corrections to be applied to  (\ref{E:modeltc}),  from $t$ to $t+N-1$, given the knowledge of the state at time $t$.
From \eqref{eq:feedbacku}, the corresponding predicted control input sequence is
$u_{j|t} = K_f x_{j|t} + v_{j|t}$, $ j=0,1,\ldots,N-1$.
By using model (\ref{E:modeltc}),
we thus obtain the predicted values of the states as linear functions of the current state $x_t$, of the predicted (to-be-determined) control sequence
$\calV_t=[v_{0|t}\tran\;\cdots\; v_{N-1|t}\tran]\tran \in\Real{Nm}$, and of the disturbance sequence $\gamma$:
\beq
x_{j|t} = A\ped{cl}^j(\theta) x_t + \Phi_j(\theta)\calV_t + \Upsilon_j(\theta) \gamma,\quad
j=1,\ldots,N,
\label{eq-xpred}
\eeq
where  $\Phi_j(\theta)$, $\Upsilon_j(\theta)$ are suitable functions of the model matrices,
$A\ped{cl}(\theta)$, $B(\theta)$ and $B_\gamma(\theta)$.
However, the predictions obtained via model (\ref{E:modeltc}) are {\em uncertain}, since
they depend on $\theta$ and on $\gamma$. In our approach, we deal with this issue by considering
a discrete set of predicted state and input trajectories, obtained for a number $M$ of randomly extracted {\em scenarios}
of $\theta$ and $\gamma$ at time $t$. More precisely, let us collect these random  parameters in $\delta = (\theta,\gamma),\,\delta\in\Delta\doteq\Theta\times\Gamma^\infty$.
As a consequence of Assumption~\ref{A:unc_bound}, we have that
 $\delta$ has a probability measure that we denote with $\prob$, which is the product measure
of $\prob_\theta$ and the  measure $\prob_\gamma^\infty$ on $\gamma$:
%\beq
$\prob = \prob_\theta \times \prob_\gamma^\infty$.
%\label{eq:probdelta}
%\eeq
Consider then
 $M$ independent extractions $\delta^{(1)}_t,\ldots,\delta^{(M)}_t$ of $\delta$, constituting the
scenarios,
where each scenario has the probability distribution $\prob$, and let $\omega_t\doteq (\delta^{(1)}_t,\ldots,\delta^{(M)}_t)$ denote the ``multisample'' of scenario extractions at time $t$.
The probability distribution on $\omega_t$ is given by $\prob^M$.
Based on the random scenarios, we obtain $M$ different state and input predictions from (\ref{eq-xpred}), namely, for
$ i=1,\ldots,M $,
\beq\label{eq-xpredi}
\ba{lll}
x_{0|t}^{(i)} &=& x_t \\
x_{j|t}^{(i)} &=& A\ped{cl}^j(\theta^{(i)}_t) x_t + \Phi_j(\theta^{(i)}_t)\calV_t + \Upsilon_j(\theta^{(i)}_t) \gamma^{(i)}_t,\\&&
\quad j=1,\ldots,N,\\
u_{j|t}^{(i)} &=& K_fx_{j|t}^{(i)}+v_{j|t},
\quad j=0,\ldots,N-1,\\
\ea
\eeq
where $(\theta^{(i)}_t,\gamma^{(i)}_t) = \delta^{(i)}_t$.
Let us now introduce  the following  cost function:
\begin{equation}\label{E:cost_single}
J(x_t,\omega_t;\calV_t)\doteq
\max\limits_{i=1,\ldots,M}\left(
\sum\limits_{j=0}^{N-1}d(x_{j|t}^{(i)},\mathbb{X}_f)+
\sum\limits_{j=0}^{N-1}v_{j|t}^T \Lambda v_{j|t}\right),
\end{equation}
where $d(x,\mathbb{X}_f)$ is the distance between $x$ and the terminal set $\mathbb{X}_f$, computed in some norm $\|\cdot\|$: $d(x,\mathbb{X}_f)\doteq\min\limits_{y\in\mathbb{X}_f}\|x-y\|$, and $\Lambda=\Lambda\tran\succ0$ is a weighting matrix chosen by the control designer.  In the following, with a slight abuse of notation, we indicate the state and input constraint sets as $\mathbb{X}(\delta),\,\mathbb{U}(\delta)$, respectively, and the related convex functions in Assumption~\ref{A:constr} as $f_X(x,\delta)$, $f_U(u,\delta)$. Moreover, we transform the hard constraints of
Assumption~\ref{A:constr} into soft ones, by introducing a slack variable $q_t\in\bR,\,q_t\geq0$. Then,
the scenario-based FHOCP is a random convex program defined as follows:
\begin{subequations}\label{E:RFHOCP}
\beq
\calP (x_t,\omega_t):\;\min\limits_{\calV_t, z_t,q_t} \; z_t+\alpha q_t  \label{E:RFHOCP_0}
\eeq
\beq
\text{subject to}\nonumber
\eeq
\beq
J(x_t,\omega_t;\calV_t) \leq z_t\label{E:RFHOCP_1}
\eeq
\beq
f_X(x_{j|t}^{(i)},\delta^{(i)}_t)-\boldsymbol{1}q_t\preceq0;\;\; j=1,\ldots,N-1,\; i=1,\ldots,M \label{E:RFHOCP_2}
\eeq
\beq
f_U(u_{j|t}^{(i)},\delta^{(i)}_t)-\boldsymbol{1}q_t\preceq0;\;\; j=0,\ldots,N-1,\;
i=1,\ldots,M \label{E:RFHOCP_3}
\eeq
\beq
f_{X_f}(x_{t+N|t}^{(i)})-\boldsymbol{1}q_t\preceq0;\;\;i=1,\ldots,M. \label{E:RFHOCP_4}
\eeq
\beq
q_t\geq0\label{E:RFHOCP_5}
\eeq
\end{subequations}
In \eqref{E:RFHOCP_0}, the weighting scalar $\alpha>0$ is chosen by the control designer, and $\boldsymbol{1}$ denotes a column vector of appropriate length, containing all ones. We denote with  $\mathcal{V}_t^*(x_t,\omega_t)=\{v^*_{0|t},\ldots,v^*_{N-1|t}\},$ $z_t^*(x_t,\omega_t)$ and $q^*_t(x_t,\omega_t)$ an optimal solution to problem \eqref{E:RFHOCP}.
\begin{remark}\label{R:soft_constraints}(Worst-case cost and constraint violation)
Due to the presence of constraint \eqref{E:RFHOCP_1}, the value $z^*_t$ is an upper bound of the worst case cost with respect to all the $M$ extracted scenarios. We thus refer to $z^*_t$ as the ``worst-case cost.'' Moreover, we note that the use of the soft constraints \eqref{E:RFHOCP_2}-\eqref{E:RFHOCP_4} imply that  problem \eqref{E:RFHOCP} is always feasible.
In particular, by using a sufficiently high value of $\alpha$ (e.g. $10^4$ times higher than the typical value of $z^*_t$), the optimal value of $q^*_t$ turns out to be negligible whenever the problem with hard constraints (i.e., with $q_t$ set a priori to zero) is feasible. Contrary, when the problem with hard constraints is not feasible, the variable $q^*_t$ provides an indication on ``how much'' some of the constraints are violated. For this reason, we refer to $q^*_t$ as the ``constraint violation'' level. Finally, we note that there is no constraint violation in \eqref{E:RFHOCP_1}, i.e. $z^*_t$ is always greater than all the cost functions corresponding to the sampled scenarios, and in particular it is always an upper bound of the distance between the state $x_t$ and the terminal set $\mathbb{X}_f$ (see \eqref{E:cost_single}). This feature is important for our convergence result in Section~\ref{S:MPC_RCP_alg}.\hfill$\blacksquare$
\end{remark}
\begin{remark}\label{R:closed_loop}(Choice of cost function and input parameterization)
Prediction of the state trajectories in a closed loop fashion is quite common in the context of robust MPC, see e.g. \cite{ScMa98,OGJM09}. In particular, we adopt here the input parameterization (\ref{eq:feedbacku}), and we optimize over the control corrections $v_{j|t},\,j=0,\ldots,N-1$, i.e.\ over  $N\,m$ decision variables. Moreover, we chose as stage cost the distance between the state and the terminal set, plus a quadratic penalty on the control correction. Indeed, these choices of control parameterization and cost function are not meant to be the sole possibility, neither the optimal, for the proposed approach.
Generalization to  other kinds of input parameterization (e.g. disturbance-feedback \cite{GoKM06,OGJM09}) and cost function (like a standard quadratic stage cost) can be done with some technical modifications in the proofs of the results reported in this note.\hfill$\blacksquare$
\end{remark}
The optimization problem $\calP(x_t,\omega_t)$ can be rewritten in a more compact standard form.
By collecting the optimization variables $(\mathcal{V}_t,z_t,q_t)$ in  vector
$s_t\in\mathbb{R}^{mN+2}$, the cost  can be expressed as  $z_t+\alpha q_t=c\tran\,s_t$, where $c=[0,\ldots,0,1,\alpha]\tran$.
Moreover, it can be noted that, for any fixed value of $\delta_t$, due to linearity of
 (\ref{eq-xpredi}), the constraints (\ref{E:RFHOCP_1})-(\ref{E:RFHOCP_5}) are convex in the decision variable $s_t$ and in the state $x_t$.
Finally, these constraints can be formally expressed compactly as
$
h(s_t,x_t,\delta_t^{(i)})\leq 0$, for all $i=i,\ldots,M$,
where
$h:\mathbb{R}^{m N+2}\times\mathbb{R}^n\times\Delta\rightarrow\mathbb{R}$ is
defined as \small$
h(s_t,x_t,\delta_t^{(i)}) \doteq
\max
\left\{ \max\limits_{j=0,\ldots, N-1}\left\{f_X(x_{j|t}^{(i)},\delta_t^{(i)})-\boldsymbol{1}q_t, f_U(u_{j|t}^{(i)},\delta_t^{(i)})-\boldsymbol{1}q_t\right\}\right.$, $\left. f_{X_f}(x_{t+N|t}^{(i)})-\boldsymbol{1}q_t,-q_t ,J(x_t,\omega_t;\mathcal{V}_t)- z_t \right\}
$\normalsize.
Notice that $h(s_t,x_t,\delta_t^{(i)}) $ is  convex  in both $s_t$ and $x_t$, since it is the point-wise maximum of convex functions.
The scenario FHOCP can hence be rewritten as
\bea
\mathcal{P}(x_t,\omega_t): & \min\limits_{s_t} c\tran s_t \label{eq_S-FHOCP}\\
\text{subject to:} &  h(s_t,x_t,\delta_t^{(i)})\leq 0 ,\; i=1,\ldots,M.\nonumber
\eea
We denote with $s^*_t(x_t,\omega_t)=
%(\mathcal{V}_t^*(x_t,\omega_t,z_t^*(x_t,\omega_t,q^*_t(x_t,\omega_t))$
(\mathcal{V}_t^*,z_t^*,q^*_t)$
 an optimal solution of $\mathcal{P}(x_t,\omega_t)$.
 Notice that, due to the way it has been defined, problem $\calP(x_t,\omega_t)$ is {\em always} feasible.
 We further assume that this problem always attains a unique optimal solution.

\subsection{Properties of the scenario FHOCP}

We now consider the following problem: suppose that, given the state $x_t$, we solve
problem $\calP(x_t,\omega_t)$. Then, we ask what is the probability  that the computed optimal control sequence $\mathcal{V}_t^*(x_t,\omega_t)$ is able to satisfy all state and input constraints over the chosen horizon, and to drive the state trajectory to the terminal set at the end of the horizon, within the computed optimal constraint violation $q_t^*$. Formally, this is the probability (with respect to $\delta$) with which  $h(s_t^*,x_t,\delta)\leq 0$,
where we notice that $h$ is now  evaluated at the optimal scenario solution $s_t^*$, and
the state and input trajectories that enter the definition of $h$ are the ``actual,'' uncertain, ones, obtained from
model (\ref{E:modeltc})  at a random $\delta=(\theta,\gamma)$. So, we define the {\em reliability} $R$ of the scenario-FHOCP as
\[
R \doteq \prob\{\delta:\, h(s_t^*,x_t,\delta)\leq 0\}.
\]
Notice  that $R\in[0,\,1]$ is itself a random variable, since it depends
on $s_t^*$, which in turn depends on the random multi-extraction of the scenarios $\omega_t$,
hence $R=R(\omega_t)$.
Indeed, for some extractions $\omega_t$ the reliability can be good (close to one), and for other
extractions it can be bad. It is therefore  critical to assess the a-priori likelihood of these two situations, that is to precisely quantify  bounds on the probability of the ``bad'' event where $\{R< p\}$, being $p$ some a-priori assigned level or desired reliability.
To this purpose, we exploit the fact that
  problem $\mathcal{P}(x_t,\omega_t)$ belongs to the class of so-called
Random Convex Programs (RCP) (see, e.g., \cite{CaCa05}-\cite{Cala10}) and, in particular, the result in Theorem 1 of \cite{CaGa08}, concerned with feasible random convex programs, applies to our context.
The following key result directly follows from Theorem 1 of \cite{CaGa08}, see also Theorem 3.3 in \cite{Cala10}.
\begin{theorem}
\label{thm-1}
Let $d=mN+2$ be the number of decision variables in problem $\calP(x_t,\omega_t)$,
let $p\in(0,\,1)$ be a given desired reliability level, let $\beta\in (0,1)$
be a given small probability level (say, $\beta=10^{-9}$),
and let $M$ be an integer such that
\beq
\Phi(p,d,M) \leq \beta,
\label{eq:beta}
\eeq
with $
\Phi(p,d,M)  \doteq \sum\limits_{j=0}^{d-1} \left(\ba{c}\!\!\! M \\ j \ea \!\!\!\right) (1-p)^j p^{M-j}.$
Then, it holds that
\beq
\prob^M\left\{
\omega_t: \,  R(\omega_t)\geq p
\right\}
\geq 1-\beta.
\label{eq:betabound}
\eeq
\hfill$\blacksquare$
\end{theorem}
\begin{remark}(Number of scenarios and ``certainty equivalence'')
The practical importance of the result in Theorem~\ref{thm-1} stems from the fact that the number $M$ of scenarios necessary to fulfill condition (\ref{eq:beta}) grows mildly with the inverse of $\beta$.
More precisely, Corollary 5.1 in \cite{Cala10} states that condition (\ref{eq:beta}) is implied by
$M \geq \frac{2}{1-p}\left( \ln \beta\inv   +d\right)$, thus $M$ grows at most logarithmically with $\beta\inv$ (tighter values of $M$ for given $\beta$ and $p$ can be obtained by inverting numerically \eqref{eq:beta}). This means in turn that the
parameter $\beta$ may be fixed by the designer to a very low level, say $\beta = 10^{-9}$, or even
$\beta = 10^{-12}$, and still the number $M$ of scenarios necessary to guarantee (\ref{eq:betabound}) remains manageable. The parameter $\beta$ hence measures the probability of the unfortunate event in which the optimal solution has  reliability smaller than the desired level $p$. If the likelihood of such an event is bounded a priori by an extremely low value, such as
$\beta = 10^{-12}$, then we may safely say that, to all practical engineering purposes, the event
$\{R(\omega_t)\geq p\}$
is the ``certain'' event. In other words, the possibility that
$\{R(\omega_t)\geq p\}$
is not satisfied by the scenario problem is so remote
that, before having any concern about it, the designer should better verify the validity of many other assumptions and approximations in the model. Such a ``certainty equivalence'' principle, which will be adopted henceforth in this paper,
essentially eliminates from consideration the ``outer'' probability level in (\ref{eq:betabound}), and states with practical certainty
(the expression ``with practical certainty'' shall be used in the rest of this note as a synonym  of ``with probability larger than $1-\beta$,'' where $\beta>0$ is some extremely small value) that
$\{R(\omega_t)\geq p\}$ holds.
This simplifies greatly  the practical application of scenario techniques,
and makes the whole approach more clear and understandable
by both theoreticians and control practitioners.\hfill$\blacksquare$
\end{remark}
The properties of the scenario FHOCP are resumed in the following proposition.
\begin{proposition}(Finite horizon robustness)\label{T:open_loop_stable}
Given the state
$x_t$ of system (\ref{E:modeltc}) at time $t$, consider the scenario problem
$\calP(x_t,\omega_t)$ as an instrument to derive a finite-horizon control sequence
$\calV_t^* = \{v^*_{0|t},\ldots,v^*_{N-1|t}\}$ to be applied to the system  (\ref{E:modeltc})
at the subsequent instants $t,t+1,\ldots, t+N-1$.
Let the number $M$ of scenarios in problem $\calP(x_t,\omega_t)$ be chosen so to satisfy (\ref{eq:beta})
for given reliability level  $p\in(0,\,1)$ and very small $\beta\in (0,1)$.
Then, with practical certainty  it holds that the computed control sequence:\\
\textbf{a)}
  steers the state of system (\ref{E:modeltc}) to the terminal set $\mathbb{X}_f$ in $N$ steps with probability at least $p$ and constraint violation $q^*_t$, i.e.: $\prob\{\delta: \, f_{X_f}(x_{t+N},\delta)-\boldsymbol{1}q_t^* \preceq 0\}\geq p$;\\
\textbf{b)} Satisfies all state  constraints with probability at least $p$ and constraint violation $q^*_t$, i.e.: $
  \prob\{\delta:\, f_{X}(x_{t+j},\delta)-\boldsymbol{1}q_t^* \preceq 0,\,\forall j\in[1,N]\}\geq p$\\
\textbf{c)}
Satisfies all input  constraints with probability at least $p$ and constraint violation $q^*_t$, i.e.: $
\prob\{\delta: \, f_{U}(u^*_{t+j},\delta )-\boldsymbol{1}q_t^*\preceq0,\,\forall j\in[0,N-1]\}\geq p.$\hfill$\blacksquare$
\end{proposition}
The proof of this result follows  immediately from Theorem~\ref{thm-1}: eq. (\ref{eq:betabound}) states  that,
with practical certainty,
the optimal solution $s^*_t$ of the scenario
problem satisfies $h(s_t^*,x_t,\delta)\leq  0$ with
probability at least $p$,
which indeed implies  that points \textbf{a)}-\textbf{c)} in the corollary hold.
\begin{remark}\label{R:feasibility}(Relationship with deterministic approaches)
In a deterministic approach to robust MPC, a problem similar to \eqref{E:RFHOCP} has to be solved \emph{for all possible values} of $\delta\in\Delta$. When the problem is convex with respect to $\delta$ (which happens, for instance, when the uncertain matrices and/or the additive disturbance belong to polytopes \cite{KoBM96,LCRM04}), deterministically robust approaches are indeed well-established and shall be preferred to the scenario approach, especially if deterministic robustness is critical in the considered application. In all other cases, deterministic approaches are generally intractable, unless the problem is manipulated  so to satisfy convexity assumptions, at the cost of higher conservativeness and reduced feasibility. In  these situations, the scenario approach proposed here is a viable alternative to deterministic techniques, since it is always convex and it can be efficiently solved also with a large number of samples, while still giving probabilistic guarantees on the robustness of the solution, as it will be shown in the example section.
Moreover, we note that, if a sufficiently  high weight $\alpha$  is used in the objective (see Remark \ref{R:soft_constraints}), then the scenario problem typically returns a negligible optimal constraint violation $q_t^*$, whenever the problem would be feasible
with $q_t$ set to zero.
Furthermore, since the scenario FHOCP has only a subset of the constraints of a corresponding
deterministically robust FHOCP, the violation level $q_t^*$ of the scenario problem will always be  lower
than the violation level of the deterministic version of the same problem
(and this fact holds independently of whether the deterministic problem can be solved numerically or not).
In any case, the  value of $q^*_t$ gives an indication on the extent of the violation of the involved constraints, which can be used, e.g., to implement supervisory control strategies and recovery actions.
\hfill$\blacksquare$
\end{remark}
The remaining part of this note is devoted to analyzing what happens when a scenario FHOCP is solved repeatedly
in time and used to control the plant in a receding-horizon fashion.
In a receding-horizon approach, which is the key feature of MPC, only the first control correction in the optimal sequence $\calV^*_t$ is applied at time $t$, and then the FHOCP is solved again at time $t+1$, by exploiting the knowledge of the state $x_{t+1}$, etc. In the next section, we propose a technique for incorporating the scenario FHOCP into a suitable receding-horizon scheme, and we derive  probabilistic guarantees of asymptotic convergence and constraint satisfaction for the resulting closed-loop system.

\section{MPC scheme based on Scenario optimization}\label{S:MPC_RCP_alg}
We here introduce a  receding-horizon implementation of a control algorithm based on
the scenario FHOCP, as described next. The notation is set as follows: ``$*$'' variables, such as $z_t^*, q^*_t, \calV_t^*=\{v_{0|t}^*,\ldots,v_{N-1|t}^*\}$,
denote the optimal solution of the scenario optimization problem $\calP(x_t,\omega_t)$ at time $t$, given $x_t$; ``$\sim$'' variables,  $\tilde z_t,\tilde q_t,\tilde \calV_t$, denote, respectively, two  scalar values and a  sequence
of $N$ vectors of dimension $m$, as defined in the algorithm below; finally plain variables, $z_t,q_t,\calV_t$, denote the running values of the variables $z,\,q$ and of the sequence
 $\calV=\{ v_{0|t},\ldots,v_{N-1|t}\}$ in the algorithm. The first entry in $\calV_t$, namely $v_{0|t}$, is the actual control correction that is applied to the system (\ref{E:modeltc}) at time $t$. The subsequence composed by the last $N-1$ elements of
$\calV_t$ is denoted with $v_{1:N-1|t}$.
We are now in position to describe the algorithm for MPC based on Scenario optimization (MPCS).
\begin{algorithm}(MPCS algorithm)\ \\
\label{Alg:MPCS}
\noindent
\textbf{(Initialization)}
Choose a desired reliability level  $p\in (0,1)$ and ``certainty equivalence'' level $\beta \in(0,1)$
(say, $\beta=10^{-9}$, or  $\beta=10^{-12}$).
Let $M$ be an integer satisfying
(\ref{eq:beta}).
Choose  $\varepsilon\in(0,1]$ (see Remark~\ref{R:comments_4} below for the meaning of $\varepsilon$ and for guidelines on its choice). Given
an initial state $x_{0}$,
extract $\omega_0$ according to $\prob^{M}$,
 solve problem
$\mathcal{P}_{M}(x_{0},\omega_0)$ and obtain
the optimal control sequence $
\mathcal{V}_{0}^* = \{v_{0|0}^*,v_{1|0}^* \ldots, v^*_{N-1|0}\},$ and the optimal objective $z^*_{0}$ and constraint violation $q^*_0$.
Set $z_0 = z^*_{0},q_0=q^*_0$, $\mathcal{V}_{0}= \mathcal{V}_{0}^*$, and apply to the system the control action
$u_0=K_f x_0+ v_{0|0}$.\\
\textbf{1)} Let $t:=t+1$, observe $x_t$, and set $\tilde{\mathcal{V}}_t= \{v_{1|t-1},\ldots,v_{N-1|t-1},0\}$ = $\{v_{1:N-1|t-1}, 0\},$\\ $\tilde{z}_t = \max\left(0,z_{t-1} -d(x_{t-1},\mathbb{X}_f)\right), \,\tilde q_t = q_{t-1}$;\\
\textbf{2)} Extract the multi-sample $\omega_t$ according to $\prob^{M}$, and solve problem
  $\mathcal{P}_{M}(x_t,\omega_t)$. Let  $(\calV_t^*,z_t^*,q^*_t)$
  be the obtained  optimal solution.\\
\textbf{3)} Evaluate the following collectively exhaustive and mutually exclusive cases:\\
\textbf{3.a)}
If   $z_t^* > \left(z_{t-1}-\varepsilon d(x_{t-1},\mathbb{X}_f)\right)$
{and}
$\tilde{z}_t<d(x_t,\mathbb{X}_f)$,
{then} set $ \calV_{t} = \tilde{\mathcal{V}}_t;\quad
  z_{t} =  0;\quad q_t=\tilde{q}_t$;\\
\textbf{3.b)} {If}   $z_t^* > \left(z_{t-1}-\varepsilon d(x_{t-1},\mathbb{X}_f)\right)$ {and} $\tilde{z}_t\geq d(x_t,\mathbb{X}_f)$,  {then} set
$\calV_{t} = \tilde{\mathcal{V}}_t;\quad
  z_{t} =  \tilde{z}_t;\quad q_t=\tilde{q}_t$;\\
\textbf{3.c)} {If}   $z_t^* \leq \left(z_{t-1}-\varepsilon d(x_{t-1},\mathbb{X}_f)\right)$, {then}
set $\calV_{t} = \calV^*_{t};\quad
  z_{t} =  {z}^*_t;\quad q_t=q^*_t$;\\
\textbf{4)} Apply the control input $u_t=K_f\,x_t+ v_{0|t}$,  then go to 1).\hfill$\blacksquare$
\end{algorithm}
\begin{remark}\label{R:comments_4} The inequality
 $z^*_t\leq \left(z_{t-1}-\varepsilon d(x_{t-1},\mathbb{X}_f)\right)$, checked at step 3) of the MPCS algorithm, can be interpreted as a verification of a required minimum improvement, in terms of worst-case cost, achieved by the newly computed optimal solution $(\mathcal{V}^*_t,z^*_t,q^*_t)$ of the scenario problem at time step $t$, with respect to the previous step. The user-defined parameter $\varepsilon\in(0,1]$ influences such a requirement:
the closer the value of $\varepsilon$ is set to 0, the more likely it is that case $z^*_t\leq \left(z_{t-1}-\varepsilon d(x_{t-1},\mathbb{X}_f)\right)$ is met, so that the MPCS algorithm relies, at each time step, on the newly computed optimal solution. Vice-versa, the closer is the value of $\varepsilon$ to 1, the more likely it is that the complementary condition $z^*_t> \left(z_{t-1}-\varepsilon d(x_{t-1},\mathbb{X}_f)\right)$ is detected, so that the MPCS algorithm employs the previously computed solution.
\hfill$\blacksquare$
\end{remark}
The next results is concerned with the guaranteed properties, in terms of constraint satisfaction and convergence to the terminal set, of the closed loop system obtained by applying Algorithm \ref{Alg:MPCS}.
\begin{theorem}(Properties of Scenario MPC)\label{T:closed_loop_stable}
Let Assumptions \ref{A:unc_bound}-\ref{A:terminal_set} be satisfied and let
$p\in(0,1)$ be a chosen reliability level. Let $v_{0|t}$, $t=0,1,\ldots$ denote the sequence of control actions produced by
Algorithm~\ref{Alg:MPCS}, and consider the
 closed loop system  obtained by applying   to (\ref{E:model}) the control law $
 u_t = K_fx_t + v_{0|t}$. Then:\\
\textbf{(a)} With practical certainty, at  all time steps $t=0,1,\ldots$,
  the probability that
  the state and input constraints are satisfied with constraint violation $q_t$ is at least  $p$, that is $\prob\{\delta: f_X(x_{t+1},\delta) -\boldsymbol{1}q_t\preceq 0 \; \cap \;   f_U(u_{t},\delta)-\boldsymbol{1}q_t\preceq 0     \} \geq p,\quad
  t=0,1,\ldots$\\
\textbf{(b)}
  Algorithm~\ref{Alg:MPCS} either:
\textbf{(i)} makes the state trajectory converge asymptotically to the terminal set, i.e. $\lim\limits_{t\to\infty} d(x_t,{\mathbb X}_f) = 0$, or \textbf{(ii)} there exists a finite time $t^*$ such that, with practical certainty,  the control sequence
  $\{v_{0|t^*}, v_{0|t^*+1},\ldots v_{0|t^*+N-1}\}$ drives the state of the  closed-loop system to the terminal set
  at time $t^*+N-1$, with probability at least $p$ and constraint violation $q_{t^*}$.\hfill$\blacksquare$
  \end{theorem}
\begin{proof}{\bf Preliminaries.}
Notice first  that, for any $t\geq 0$, if $x_t\in\mathbb{X}_f$ then the optimal solution to problem $\calP(x_t,\omega_t)$
is $(\mathcal{V}^*_t,z^*_t,q^*_t)=(0,0,0)$, since the terminal control law (i.e., with $\mathcal{V}_t=0$) is able to keep the predicted state trajectory in the terminal set while satisfying all constraints.
Also, if $z^*_t=0$ is the optimal objective of problem  $\calP(x_t,\omega_t)$, then $x_t\in\mathbb{X}_f$, since $z^*_t$ is an upper bound of $d(x_t,\mathbb{X}_f)$ (see also Remark 3.1), therefore $z^*_t=0\iff x_t\in\mathbb{X}_f.$ Let then $x_t\not \in \mathbb{X}_f$.\\
\textbf{Proof of statement (a).}
At time $t=0$, Proposition~\ref{T:open_loop_stable} guarantees with practical certainty
that the first control correction satisfies the constraints
on $u_0$ and $x_1$  with probability no less than $p$ and constraint violation $q_t=q^*_t$.
At any generic time step $t\geq 1$, the variables $(\tilde{\mathcal{V}}_t,\tilde{z}_t,\tilde{q}_t)$
% (\ref{E:generic_tilde_V})-(\ref{E:generic_tilde_z})
%
are computed. Then, two cases may occur. If
$z^*_t\leq \left(z_{t-1}-\varepsilon d(x_{t-1},\mathbb{X}_f)\right)$, then case
3.c) is detected, and
the first element $v^*_{0|t}$ of the optimal sequence  $\mathcal{V}^*_t$ is applied to the system.
Being this sequence the solution of a scenario optimization problem, with practical certainty the probability of satisfying state and input constraints is no less than $p$, with constraint violation $q_t=q^*_t$. If, on the other hand,
$z^*_t> \left(z_{t-1}-\varepsilon d(x_{t-1},\mathbb{X}_f)\right)$, then
we are either in case 3.a) or 3.b), and in both cases
the element $v^*_{k|t-k},$ for some $k\in[1,N-1]$, is applied to the system. Being this value part of the solution sequence $\mathcal{V}^*_{t-k}$, with corresponding constraint violation $q^*_{t-k}$, again the probability of satisfying state and input constraints is no less than $p$, with constraint violation $q_t=\tilde{q}_t=q^*_{t-k}$.
Thus, in any case, with practical certainty, at each time step the MPCS algorithm guarantees satisfaction of state and input constraints with probability no less than~$p$ and constraint violation $q_t$.\\
{\bf Proof of statement (b).}
Each run of Algorithm~4.1 may have one of two possible behaviors,
depending on
whether or not  there exists a finite time
$t > 0$ such that $
z_t^* > \left(z_{t-1}-\varepsilon d(x_{t-1},\mathbb{X}_f)\right) \mbox{ and } \tilde{z}_t<d(x_t,\mathbb{X}_f),$ that is, whether or not  the situation in step 3.a is ever satisfied.
We then name  $\calA$
the situation when condition in step 3.a is met at some finite $t>0$, and $\bar \calA$ the
complementary situation when this condition is not satisfied at any finite time, that is when $z_t^* \leq \left(z_{t-1}-\varepsilon d(\xi_{t-1},\mathbb{X}_f)\right) \mbox{ or } \tilde{z}_t\geq d(\xi_t,\mathbb{X}_f)$ holds for all $t> 0$.\\
\textbf{I.} Let us first consider the situation of case $\bar \calA$.
Consider a generic time  $t$.
At step 3) of the MPCS algorithm, if
$z^*_t> \left(z_{t-1}-\varepsilon d(x_{t-1},\mathbb{X}_f)\right)$, then, since it is assumed that we are in situation $\bar\calA$, it must hold that $\tilde{z}_t\geq d(x_t,\mathbb{X}_f)$, thus case 3.b) occurs,
 and the values $\mathcal{V}_t=\tilde{\mathcal{V}}_t$ and $z_t=\tilde{z}_t$ are set.
 Now, recalling that $\tilde{z}_t=\max(0,z_{t-1}-d(x_{t-1},\mathbb{X}_f))$, two cases may occur:
 either $\tilde{z}_t=0$ or $\tilde{z}_t=z_{t-1}-d(x_{t-1},\mathbb{X}_f)>0$. If $\tilde{z}_t=0$, we have $0=\tilde{z}_t\geq d(x_t,\mathbb{X}_f),$ %(recalling that  $z^*_t \geq d(x_t,\mathbb{X}_f)\geq 0,\,\forall t=0,1,\ldots$, see Remark \ref{R:soft_constraints}),
 i.e. $d(x_t,\mathbb{X}_f)=0$, which would imply that the terminal set has been reached. Otherwise, if $\tilde{z}_t=z_{t-1}-d(x_{t-1},\mathbb{X}_f)>0$, then we have:
      \begin{equation}\label{E:dist_bnd_6}
      \begin{array}{l}
      z_t = \tilde{z}_t \geq d(x_t,\mathbb{X}_f)\geq0,
      \end{array}
      \end{equation}
      and $z_t-z_{t-1}= \tilde{z}_t-z_{t-1}=z_{t-1}-d(x_{t-1},\mathbb{X}_f)-z_{t-1}=-d(x_{t-1},\mathbb{X}_f)$. Thus,
      \begin{equation}\label{E:delta_z_2}
      \begin{array}{l}
      z_t-z_{t-1}\leq-\varepsilon d(x_{t-1},\mathbb{X}_f),\,\forall x_{t-1}\not\in\mathbb{X}_f,
      \end{array}
      \end{equation}
      \begin{equation}\label{E:delta_z_3}
      \begin{array}{l}
      \text{and }z_t-z_{t-1}=0\iff x_{t-1}\in\mathbb{X}_f.
      \end{array}
      \end{equation}
      On the other hand, if at step 3) of the MPCS algorithm it happens that
      $z^*_t\leq \left(z_{t-1}-\varepsilon d(x_{t-1},\mathbb{X}_f)\right)$, then case 3.c)
      occurs, and the optimal values $\mathcal{V}^*_t$ and $z^*_t$ are retained, i.e.
      $z_t = z^*_t$, $\calV_t = \mathcal{V}^*_t$.
      In this case, it is straightforward to note that equations (\ref{E:dist_bnd_6})-(\ref{E:delta_z_3}) still hold true. The same reasoning can be repeated
      for any time step, as long as the case
      $z_t^* \leq \left(z_{t-1}-\varepsilon d(x_{t-1},\mathbb{X}_f)\right) \mbox{ or }\tilde{z}_t\geq d(x_t,\mathbb{X}_f)$ holds true as assumed, so that we can conclude that the variable $z_t$ enjoys the following properties:
      \begin{equation}\label{E:Lyap_prop}
        \begin{array}{l}
        z_t\geq d(x_{t},\mathbb{X}_f)\geq0,\,\forall t\geq0\\
        z_t=0\iff x_t\in\mathbb{X}_f\\
        z_{t+1}-z_t\leq-\varepsilon d(x_t,\mathbb{X}_f),\,\,\forall x_t\not\in\mathbb{X}_f,\,\forall t\geq0\\
        z_{t+1}-z_t=0\iff x_t\in\mathbb{X}_f.
        \end{array}
        \end{equation}
        Properties (\ref{E:Lyap_prop}) are sufficient to prove convergence of the state $x_t$ to the set $\mathbb{X}_f$:
        \[
        \begin{array}{l}
        0\leq\lim\limits_{t\rightarrow\infty}d(x_{t},\mathbb{X}_f)\leq\lim\limits_{t\rightarrow\infty}z_t=0,        \Rightarrow\lim\limits_{t\rightarrow\infty}d(x_{t},\mathbb{X}_f)=0 .
        \end{array}
        \]
      Therefore, we obtain that in case $\bar\calA$
the MPCS algorithm guarantees that
$\lim\limits_{t\to\infty} d(x_{t},\mathbb{X}_f)=0$.\\
\textbf{II.} Let us next analyze what happens in case  $\calA$.
Let $\bar t > 0$ be the time instant at which the
case $z_t^* > \left(z_{t-1}-\varepsilon d(x_{t-1},\mathbb{X}_f)\right)  \mbox{ and } \tilde{z}_t<d(x_t,\mathbb{X}_f)$ is met for the first time, and let
 $t^*<\tbar$ be the last time at which case
 $z^*_t\leq \left(z_{t-1}-\varepsilon d(x_{t-1},\mathbb{X}_f)\right)$ was satisfied, that is the last time previous to $\tbar$
 when an optimal command sequence was retained, together with its constraint violation $q^*_t$, according to case 3.c) of Algorithm 4.1;
 let
 $\ell = \tbar - t^*\geq 1$.
According to step 3.a) of the MPCS algorithm, we set
\beq
\mathcal{V}_\tbar=\tilde{\mathcal{V}}_\tbar,\,\,z_{\tbar}=0,\,\,q_{\tbar}=\tilde{q}_{\tbar}.
\label{E:dist_bnd_5}
\eeq
Thus, at step 4) of the algorithm, the control move
$u_\tbar =K_fx_\tbar +v_{0|\tbar}$ is applied to the system at time $\bar t$,
where  $v_{0|\tbar} = v^*_{\ell|t^*}$, i.e., $v_{0|\tbar}$ is
the optimal correction predicted for time $t^*+\ell=\tbar$, computed at time $t^*$. At time step $t=\tbar+1$, the state variable $x_{\tbar+1}$ is observed and
 $(\tilde{V}_{\tbar+1},\tilde{z}_{\tbar+1},\tilde{q}_{\tbar+1})$ are computed as $\tilde{z}_{\tbar+1}=\max(0,z_\tbar - d(x_\tbar,\mathbb{X}_f)),\,\tilde{q}_{\tbar+1}=q_\tbar,\,\tilde{\mathcal{V}}_{\tbar+1}=
      \{v_{1:N-1|\tbar}, 0\} =
      \{v_{\ell+1|t^*}^{*},\,
      v_{\ell+2|t^*}^{*},\ldots,v_{N-1 |t^*}^{*},0,\ldots,0\}.$ Since (\ref{E:dist_bnd_5}) holds, it must be $\tilde{z}_{\tbar+1}=0$.
     Then,  $z^*_{\tbar+1}$, $q^*_{\tbar+1}$ and $\calV^*_{\tbar+1}$ are computed at step 2), and we notice that,
     by definition, $z^*_{\tbar+1}\geq 0$.
     Therefore, at step 3) of the algorithm either {\em (i)} case 3.a)
     \{$z_{\tbar+1}^* > \left(z_\tbar-\varepsilon d(x_{\tbar},\mathbb{X}_f)\right)
     \mbox{ and } \tilde{z}_{\tbar+1}<d(x_{\tbar+1},\mathbb{X}_f)$\}
      is detected again, or {\em (ii)}  one of cases 3.b) or 3.c) are detected, which would imply, respectively,
      $0=\tilde z_{\tbar+1}\geq d(x_{\tbar+1},\mathbb{X}_f)$, or $0\leq d(x_{\tbar+1},\mathbb{X}_f)\leq z^*_{\tbar+1}= \tilde z_{\tbar+1}=0$. Hence (in either case) $x_{\tbar+1}\in\mathbb{X}_f$,  %see (\ref{E:iff_terminal_set}),
     so that
     convergence to the terminal set would be achieved.
     Consider then case {\em (i)}:
     the values $\mathcal{V}_{\tbar+1}=\tilde{\mathcal{V}}_{\tbar+1}$,
      $z_{\tbar+1}=0$ and $q_{\tbar+1}=q_{\tbar}$ are set in the algorithm, and the control move
      $u_{\tbar+1}=K_fx_{\tbar+1}+v^*_{{\ell+1}|{t^*}}$ is applied to the system.
     Now, the same circumstances actually reproduce for all time steps $t=\tbar+k$, $k\geq 0$, so the algorithm is such that
   the optimal input sequence $\mathcal{V}^*_{t^*}$, computed at time $t^*$ by solving a scenario FHOCP,
   is the one actually next  applied to the system, and the related constraint violation $q^*_t$ is retained for  all $t\geq t^*$.
   Thus, in case  $\calA$, there exists a finite time $t^*$ such that the  sequence
   $\mathcal{V}^*_{t^*}$ is  applied to the system for all subsequent instants $t= t^*+k$, $k=0,\ldots,N-1$. Now, the sequence $\mathcal{V}^*_{t^*}$ is the result of the solution of the scenario-FHOCP
   $\calP(x_{t^*},\omega_{t^*})$, and Proposition~\ref{T:open_loop_stable} states that, with practical certainty,
   we have $R(\omega_{t^*})\geq p$, where $R$ is the reliability defined in
   Section III-A of the paper, which means that $\prob\{\delta: h(s^*_{t^*}, x_{t^*},\delta) \leq 0\} \geq p$. Therefore, in the situation $\calA$, there exists a finite time $t^*$ at which an optimal control sequence
   is computed by solving a scenario-FHOPC and next applied to the actual system for the subsequent $N$ time instants: we can hence claim with practical certainty this sequence will satisfy the problem constraints and reach the terminal set within the time window from $t^*$ to $t^*+N$, with probability at least $p$ and constraint violation $q^*_t$.\end{proof}

\section{Numerical example}\label{S:example}
We consider the system (\ref{E:model}) with: \small$
A(\theta)=
\left[\begin{array}{cc}
1+\theta_1 & \dfrac{1}{1+\theta_1}\\0.1\sin(\theta_4) & 1+\theta_2
\end{array}\right]$, $
B(\theta)=\left[\begin{array}{c}
0.3\arctan(\theta_5)\\\dfrac{1}{1+\theta_3}
\end{array}\right]$, $
B_\gamma=\left[\begin{array}{cc}
1 & 0\\0 & 1
\end{array}\right]$. \normalsize
Each of the random parameters $\theta_1,\theta_2,\theta_3$ is uniformly distributed in the interval [-0.1,\,0.1], while $\theta_4,\theta_5$ are distributed according to Gaussian distributions with zero mean and unit variance. Moreover, the disturbance $\gamma_t\in\mathbb{R}^2$ is computed as follows: \small
\[
\gamma_t=\left\{
\ba{l}
\left[\eta_1,\,\min\left(\eta_2,\frac{1}{3}\left(\frac{1}{100\left(3\eta_1+0.05\right)}-0.05\right)\right)\right]\tran\text{ if }\eta_0\geq0.5\\
\left[0.05\cos\left(\eta_3\right),\,\eta_5\right]\tran \text{ if }\eta_0<0.5\\
\ea
\right.
\]\normalsize
where\small \[\eta_5=\max\left(\min\left(\eta_4\sin\left(\frac{\pi}{4}\right),0.05|\sin\left(\eta_3\right)|\right),
-0.05|\sin\left(\eta_3\right)|\right)
 \]\normalsize and $\eta_0\in[0,1],\,\eta_1\in[0,0.05],\,\eta_2\in[0,0.05],\,\eta_3\in[\frac{3}{4}\pi,\frac{5}{4}\pi],\,\eta_4\in[-0.05,0.05]$ are uniformly distributed random variables. We also consider the following constraints on the input and state variables: $
\mathbb{X}(\theta)\doteq\left\{x\in\mathbb{R}^2:\|x\|_\infty\leq\overline{x}(\theta)\right\},
\quad
\mathbb{U}(\theta)\doteq\left\{u\in\mathbb{R}:|u|\leq\overline{u}(\theta)\right\},
$
 where $\overline{u}(\theta)=5/(1+\theta_6\sin(\theta_7))$, $\overline{x}(\theta)=[10/(1-\theta_6\sin(\theta_7)),\,10/(1+\theta_6\cos(\theta_7))]\tran$, $\theta_6\in[-0.05,0.05]$ is uniformly distributed, finally $\theta_7$ is distributed according to a Gaussian distribution with zero mean and unit variance. Note that the system matrices and constraints are a nonlinear functions of the uncertain parameters, and the disturbance belongs to a non-convex, disconnected set. Hence, no existing technique for robust MPC can be directly applied in this example. By using the results in \cite{Mao03,PNTN06} (see Section \ref{S:problem}), we computed the terminal control law $K_f$ and terminal set $\mathbb{X}_f$ satisfying Assumption \ref{A:terminal_set}: $
 K_f =[-0.4686,\,-1.4221],\quad
 \mathbb{X}_f=\{x\in\bR^2:x\tran Q_f x\leq1\},$ \small$
 Q_f=\left[\begin{array}{cc}
    0.0539 &   0.0724\\
    0.0724 &  0.1724
\end{array}\right]$\normalsize.
We designed the MPCS law with $N=10$ and $\Lambda=1$. We set $\beta=10^{-9}$, and considered different values of $p$.
In order to estimate the probability with which the MPCS control law and the finite horizon sequence $\mathcal{V}^*_0$ satisfy the constraints and drive the state to the terminal set, and to compare it with the bound $p$, we carried out $N_\text{trials}=100,000$ Monte Carlo simulations, starting from the same initial state $x_0=[5,\,2.75]^T$.
We note that this initial condition is not feasible for the deterministic counterpart of the scenario problem, hence for some extractions of $\omega_0$ the constraint violation $q_0^*$ is not negligible. In each one of these simulations,  Algorithm \ref{Alg:MPCS} has been applied and the probability of success $\hat{p}$ has been estimated as $
\hat{p}=\dfrac{N_\text{trials}-N_\text{failures}}{N_\text{trials}}$, where $N_\text{failures}$ is the number of simulations in which some of the constraints were not satisfied. Moreover, the finite horizon solution was also tested, to check the result of Proposition \ref{T:open_loop_stable}. In particular, for the finite horizon sequence $\mathcal{V}^*_0$, the considered constraints were the state and input constraints $\mathbb{X}(\theta),\,\mathbb{U}(\theta)$, and the terminal set constraint $x_{N}\in\mathbb{X}_f$, while for the receding horizon implementation the latter constraint was replaced with $x_{N+10}\in\mathbb{X}_f$, in order to approximately take into account the asymptotic stability result of Theorem \ref{T:closed_loop_stable} (b)-(i). We evaluated these constraints as hard constraints, i.e. with zero constraint violation. The values of $\hat{p}$ for the finite horizon simulations and for the receding horizon ones are indicated as $\hat{p}^\text{FH}$ and $\hat{p}^\text{RH}$, respectively.
\begin{table}[hbt]
  \centering
  \caption{Numerical example. Estimates $\hat{p}$ of the probability of success for different values of $p$, with $\beta=10^{-9}$.}\label{T:sim_priori}
\begin{tabular}{lcc}\hline
   $M (p)$          &22 (0.05)         &42 (0.30)\\\hline
    &$\hat{p}^\text{FH}=0.885,\,\hat{p}^\text{RH}=0.921$       &$\hat{p}^\text{FH}=0.901,\,\hat{p}^\text{RH}=0.943$\\\hline\hline
   $M (p)$          &95 (0.60)     &890 (0.95)    \\\hline
       &$\hat{p}^\text{FH}=0.923,\,\hat{p}^\text{RH}=0.963$
       &$\hat{p}^\text{FH}=0.993,\,\hat{p}^\text{RH}=0.999$  \\\hline
\end{tabular}
\end{table}
The obtained results are reported in Table \ref{T:sim_priori} for values of $p=0.05,\,0.3,\,0.6,\,0.95$. The corresponding values of $M$ are also reported in the Table. It can be noted that in all cases the values of $\hat{p}^\text{FH},\,\hat{p}^\text{RH}$ are higher than the corresponding $p$, in accordance with the theoretical results of Sections \ref{SS:randomized-control}-\ref{S:MPC_RCP_alg}. Moreover, the values of $\hat{p}^\text{FH}$ are lower than those of $\hat{p}^\text{RH}$. The estimated probabilities of success $\hat{p}^\text{FH},\,\hat{p}^\text{RH}$ are quite good already with low values of $M$: in the finite horizon case, the reason of such a good result is mainly the presence of the terminal control law $K_f$, which has already some degree of robustness, while in the receding horizon case, a higher robustness derives from the iterative re-optimization of the corrective control sequence $\mathcal{V}^*_t$. It is worth to notice that the value of $M$ does not depend on the dimension of the state variable or of the uncertainty/disturbance variables;  it only depends on the chosen values of $p,\,\beta$ and on the number of decision variables in the scenario FHOCP, i.e., the number $m$ of inputs multiplied by the control horizon $N$, plus the slack variables $z$ and $q$. However, the number of constraints embedded in $h(s,x,\delta)$ depends linearly on $n$, $m$ and $N$, so that the growth of the overall number of constraints in the scenario problem, for fixed values of $p$ and $\beta$, is $\sim(n\cdot m^2\cdot N^2)$, i.e. quadratic with respect to the control horizon.

\bibliographystyle{IEEEtran}

\begin{thebibliography}{10}

\bibitem{MRRS00}
D.~Q. Mayne, J.~B. Rawlings, C.~V. Rao, and P.~Scokaert, ``Constrained model
  predictive control: stability and optimality,'' \emph{Automatica}, vol.~36,
  pp. 789--814, 2000.

\bibitem{KoBM96}
M.~Kothare, V.~Balakhrisna, and M.~Morari, ``Robust constrained model
  predictive control using linear matrix inequalities,'' \emph{Automatica},
  vol.~32, no.~10, pp. 1361--2378, 1996.

\bibitem{WaKo03}
Z.~Wan and M.~Kothare, ``An efficient off-line formulation of robust model
  predictive control using linear matrix inequalities,'' \emph{Automatica},
  vol.~39, pp. 837--846, 2003.

\bibitem{MuBF06}
D.~{Mu\~{n}oz de la Pe\~{n}a}, A.~Bemporad, and C.~Filippi, ``Robust explicit
  {MPC} based on approximate multiparametric convex programming,'' \emph{IEEE
  Transactions on Automatic Control}, vol.~51, no.~8, pp. 1399--1403, August
  2006.

\bibitem{ScMa98}
P.~Scokaert and D.~Mayne, ``Min-max feedback model predictive control for
  constrained linear systems,'' \emph{IEEE Transactions on Automatic Control},
  vol.~43, no.~8, pp. 1136--1142, 1998.

\bibitem{ChRZ01}
L.~Chisci, J.~A. Rossiter, and G.~Zappa, ``Systems with persistent
  disturbances: predictive control with restricted constraints,''
  \emph{Automatica}, vol.~37, pp. 1019--1028, 2004.

\bibitem{LCRM04}
W.~Langson, I.~Chryssochoos, S.~Rakovi\'{c}, and D.~Mayne, ``Robust model
  predictive control using tubes,'' \emph{Automatica}, vol.~40, pp. 125--133,
  2004.

\bibitem{GoKM06}
P.~Goulart, E.~Kerrigan, and J.~Maciejowski, ``Optimization over state feedback
  policies for robust control with constraints,'' \emph{Automatica}, vol.~42,
  pp. 523--533, 2006.

\bibitem{BeBe09}
D.~Bernardini and A.~Bemporad, ``Scenario-based model predictive control of
  stochastic constrained linear systems,'' in \emph{Joint 48$^\text{th}$ IEEE
  Conference on Decision and Control and 28$^\text{th}$ Chinese Control
  Conference}, Shanghai, P.R. China, 2009, pp. 6333--6338.

\bibitem{PrSu09}
J.~Primbs and C.~Sung, ``Stochastic receding horizon control of constrained
  linear systems with state and control multiplicative noise,'' \emph{IEEE
  Transactions on Automatic Control}, vol.~54, no.~2, pp. 221--230, 2009.

\bibitem{CKRC11}
M.~Cannon, B.~Kouvaritakis, S.~Rakovi\'{c}, and Q.~Cheng, ``Stochastic tubes in
  model predictive control with probabilistic constraints,'' \emph{IEEE
  Transactions on Automatic Control}, vol.~56, no.~1, pp. 194--200, 2011.

\bibitem{CACJ??}
E.~Cinquemani, M.~Agarwal, D.~Chatterjee, and J.~Lygeros, ``On convex problems
  in chance-constrained stochastic model predictive control,''
  \emph{Automatica}, vol.~47, no.~9, pp. 2082--2087, 2011.

\bibitem{TeCD05}
R.~Tempo, G.~Calafiore, and~F. Dabbene, {\em Randomized Algorithms for
Analysis and Control of Uncertain Systems.} Springer, 2005.

\bibitem{Batina2004}
I.~Batina, ``Model predictive control for stochastic systems by randomized
  algorithms,'' Ph.D. dissertation, Technische Universiteit Eindhoven, 2004.

\bibitem{KaML09}
N.~Kantas, J.~Maciejowski, and A.~Lecchini-Visintini, ``Sequential Monte Carlo
  for model predictive control,'' in \emph{Nonlinear Model Predictive Control},
  ser. Lecture Notes in Control and Information Sciences, L.~Magni,
  D.~Raimondo, and F.~Allg\"{o}wer, Eds.\hskip 1em plus 0.5em minus 0.4em\relax
  Springer Berlin / Heidelberg, 2009, vol. 384, pp. 263--273.

\bibitem{BOBW10}
L.~Blackmore, M.~Ono, A.~Bektassov, and B.~C. Williams, ``A probabilistic
  particle-control approximation of chance-constrained stochastic predictive
  control,'' \emph{IEEE Transactions on Robotics}, vol.~26, pp. 502--517, 2010.

\bibitem{CaCa05}
G.~Calafiore and M.~Campi, ``Uncertain convex programs: Randomized solutions
  and confidence levels,'' \emph{Mathematical Programming}, vol. 102, no.~1,
  pp. 25--46, 2005.

\bibitem{CaCa06}
------, ``The scenario approach to robust control design,'' \emph{IEEE
  Transactions on Automatic Control}, vol.~51, pp. 742--753, 2006.

\bibitem{CaGa08}
M.~Campi and S.~Garatti, ``The exact feasibility of randomized solutions of
  uncertain convex programs,'' \emph{SIAM Journal on Optimization}, vol.~19,
  no.~3, pp. 1211--1230, 2008.

\bibitem{Cala10}
G.~Calafiore, ``Random convex programs,'' \emph{SIAM Journal on Optimization},
  vol.~20, pp. 3427--3464, 2010.

\bibitem{Mao03}
W.~Mao, ``Robust stabilization of uncertain time-varying discrete systems and
  comments on ``an improved approach for constrained robust model predictive
  control,'' \emph{Automatica}, vol.~39, pp. 1109--1112, 2003.

\bibitem{PNTN06}
B.~Polyak, A.~Nazin, M.~Topunov, and S.~Nazin, ``Rejection of bounded
  disturbances via invariant ellipsoids technique,'' in \emph{Proceedings of
  the 45th IEEE Conference on Decision and Control}, San Diego, CA, 2006, pp.
  1429 -- 1434.

\bibitem{BlMi08}
F.~Blanchini and S.~Miani, ``Set-Theoretic Methods in Control,''
\emph{Series: Systems \& Control: Foundations \& Applications}, Springer, 2008

\bibitem{RaBa10}
S.~V. Rakovi\'{c} and M.~Bari\'{c}, ``Parameterized robust control invariant
  sets for linear systems: Theoretical advances and computational remarks,''
  \emph{IEEE Transactions on Automatic Control}, vol.~55, no.~7, pp.
  1599--1614, 2010.

\bibitem{OGJM09}
F.~Oldewurtel, R.~Gondhalekar, C.~N. Jones, and M.~Morari, ``Blocking
  parameterizations for improving the computational tractability of affine
  disturbance feedback MPC problems,'' in \emph{Proceedings of the
  48$^\text{th}$ IEEE Conference on Decision and Control and 28$^\text{th}$
  Chinese Control Conference.}, Shanghai, China, 2009, pp. 7381--7386.

\end{thebibliography}

\end{document}